\documentclass[a4paper,USenglish]{lipics-v2019}

\usepackage{fancyhdr}
\usepackage{vmargin}
\usepackage{appendix}
\usepackage{listings}
\usepackage{algorithmic}
\usepackage{amssymb}
\usepackage{pifont}
\usepackage{multicol}
\usepackage{multirow}
\usepackage{booktabs}
\usepackage[ruled,vlined]{algorithm2e}
\usepackage{svg}

\newcommand{\ignore}[1]{}

\newcommand{\Nat}{\mathbb{N}}

\newcommand{\cQ}{\mathcal Q}
\newcommand{\cF}{\mathcal F}
\newcommand{\cM}{\mathcal M}
\newcommand{\cS}{\mathcal S}
\newcommand{\cG}{\mathcal G}

\title{Relaxed Reliable Broadcast for Decentralized Trust}

\author{João Paulo Bezerra}{LTCI, Télécom Paris, Institut Polytechnique de Paris}{joaopaulo.bezerra@telecom-paris.fr}{}{}

\author{Petr Kuznetsov}{LTCI, Télécom Paris, Institut Polytechnique de Paris}{petr.kuznetsov@telecom-paris.fr}{}{}

\author{Alice Koroleva}{ITMO University}{alicekoroleva239@gmail.com}{}{}

\Copyright{J. Paulo Bezerra, P. Kuznetsov, A. Koroleva}

\authorrunning{J. Paulo Bezerra, P. Kuznetsov, A. Koroleva}

\nolinenumbers

\begin{document}





\maketitle

\begin{abstract}
Reliable broadcast is a fundamental primitive,
widely used as a building block for data replication in distributed systems. 
Informally, it ensures that system members deliver the same values, even in the presence of equivocating Byzantine participants.   
Classical broadcast protocols are based on centralized (globally known) \emph{trust assumptions} defined via sets of participants (\emph{quorums}) that are likely not to fail in system executions.
In this paper, we consider the reliable broadcast abstraction in \emph{decentralized trust} settings, where every system participant chooses its quorums locally. 
We introduce a class of relaxed reliable broadcast abstractions that perfectly match these settings.
We then describe a broadcast protocol that achieves \emph{optimal consistency}, measured as the maximal number of different values from the same source that the system members may deliver. 
In particular, we establish how this optimal consistency is related to parameters of a \emph{graph representation} of decentralized trust assumptions.

\end{abstract}

\section{Introduction}

 Reliable broadcast is widely used for replicating data in countless applications: storage systems, state-machine replication, cryptocurrencies, etc. 
 Intuitively, a reliable broadcast protocol allows a system member (the \emph{source}) to broadcast a value, and ensures that correct system members agree on the value they deliver, despite arbitrary (\emph{Byzantine}~\cite{lamport1982byzantine}) behavior of some of them (including the source) and asynchrony of the underlying network.
 More precisely, the primitive ensures the following properties:
 \begin{itemize}
    \item (Validity) If the source broadcasts $m$, then every correct process eventually delivers $m$.
    \item (Consistency) If correct processes $p$ and $q$ deliver, respectively, $m$ and $m'$, then $m = m'$.
    \item (Integrity) Every correct process delivers at most one value, and, if the source is correct, only if it previously broadcast it.
    \item (Totality) If a correct process delivers a value, then all correct processes eventually deliver some value.
\end{itemize}
 Classical reliable broadcast algorithms, starting from Bracha's broadcast~\cite{bracha1987asynchronous},  
assume  that ``enough'' system members remain correct. 
In the \emph{uniform} fault model, where processes
fail with equal probabilities, independently of each other, this assumption implies that only less than one third of processes can fail~\cite{bracha1985asynchronous}. 

More general fault models can be captured via \textit{quorum systems}~\cite{malkhi1998byzantine}. 
Formally, a quorum system is a collection of member subsets (\emph{quorums}).
Every two quorums must have at least one correct process in common, and in every system run, at least one quorum must only contain correct processes. 

Intuitively, quorums encapsulate \emph{trust} the system members express to each other.
Every quorum can act on behalf of the whole system: before delivering a value from a potentially Byzantine source, one should make sure that a quorum of system members have \emph{acknowledged} the value.
Conventionally, these trust assumptions are centralized: all participants share the same quorum system. 

In a large-scale distributed system, it might be, however, difficult to expect that all participants come to the same trust assumptions.
It could be more realistic to resort to \emph{decentralized trust} assumptions by allowing each participant to individually choose its quorum system.  

Damg{\aa}rd et al.~\cite{damgaard2007secure} appear to be the first to consider the decentralized trust setting.
They focused on solving broadcast, verifiable secret sharing and multiparty computation, assuming \emph{synchronous} communication.
Recently, the approach has found promising applications in the field of cryptocurrencies, with the advent of Ripple~\cite{schwartz2014ripple} and Stellar~\cite{mazieres2015stellar} that were conceived as \emph{open} payment systems, alternatives to \textit{proof-of-work}-based protocols~\cite{nakamoto2008bitcoin,ethereum}.
In particular, Stellar and its followups~\cite{garcia2018federated,garcia2019deconstructing} determine necessary and sufficient conditions on the individual quorum systems, so that a well-defined subset of participants can solve the problems of consensus and reliable broadcast.

In this paper, we propose to take a more general, and arguably more realistic, perspective on decentralized trust. 
Instead of determining the weakest model in which a given problem can be solved, we rather focus on determining the strongest problem that can be solved in a given model.    
Indeed, we might have to accept that individual trust assumptions are chosen by the users independently and may turn out to be poorly justified. 
Furthermore, as in the real world, where a national economy typically exhibits strong internal trust but may distrust other national economies,  the system may have multiple mutually distrusting ``trust clusters''.
Therefore it is important to characterize the class of problems that can be solved, given specific decentralized trust assumptions.  

To this purpose, we introduce a class of \emph{relaxed} broadcast abstractions, \emph{$k$-consistent reliable broadcast} ($k$-CRB), $k\in\Nat$, that appear to match systems with decentralized trust.
If the source of the broadcast value is correct, then 
$k$-CRB ensures the safety properties of reliable broadcast (Consistency and Integrity).    
However, if the source is Byzantine, then Consistency is relaxed so that correct participants are allowed to deliver \emph{up to $k$} distinct values.
Moreover, we also refine the Totality property: if a correct process delivers a value, then every \emph{live} correct process\footnote{A process is live in a given execution if at least one of its quorums consists of correct processes only. Intuitively, we may not be able to guarantee liveness to the processes that, though correct, do not ``trust the right guys''.} eventually delivers a value \emph{or} produces an irrefutable evidence that the source is Byzantine. 
In other words, we introduce the \emph{accountability} feature to the broadcast abstraction: either the live correct processes agree on the values broadcast by the source or detect its misbehavior.  

The question now is how to determine the smallest $k$ such that $k$-CRB can be implemented given specific decentralized trust assumptions.  
We show that the trust assumptions induce a collection of \emph{trust graphs}.
It turns out that the optimal $k$ is then precisely the  size of the largest \emph{maximum independent set} over this collection of graphs.

Reliable broadcast is a principal building block for higher-order abstractions, such as state-machine replication~\cite{pbft} and asset transfer~\cite{cons-crypto,astro-dsn}. 
We see this work as the first step towards determining the strongest relaxed variants of these abstractions that can be implemented in decentralized-trust settings.   

The rest of the paper is organized as follows. In Section~\ref{sec:model}, we present our system model. In Section~\ref{sec:b_protocol}, we recall definitions of classical broadcast primitives and introduce a relaxed variant adjusted for decentralized trust settings---\textit{k-consistent broadcast} ($k$-CB). 
Section~\ref{sec:lower} introduces graph representations of trust assumptions, which are used to establish a lower bound on parameter $k$ of relaxed broadcast. 
In Section~\ref{sec:algo}, we introduce a stronger primitive, $k$-consistent reliable broadcast ($k$-CRB) and describe its  implementation. 
Finally, we discuss related work in Section~\ref{sec:related}, and we draw our conclusions in Section~\ref{sec:conclusion}.

\section{System Model}
\label{sec:model}

\subsection{Processes}

A system is composed of a set of \emph{processes} $\Pi = \{p_1,...,p_n\}$. 
Every process is assigned an \emph{algorithm} (we also say \emph{protocol}), an automaton defined as a set of possible \textit{states} (including the \textit{initial state}), a set of \textit{events} it can produce and a transition function that maps each state to a corresponding new state. 
An event is either an $input$ (a call operation from the application or a message received from another process) or an $output$ (a response to an application call or a message sent to another process); \textit{send} and \textit{receive} denote events involving communication between processes.

\subsection{Executions and failures}

A \textit{configuration} $C$ is the collection of states of all processes. In addition, $C^0$ is used to denote a special configuration where processes are in their initial states.
An \textit{execution} (or a \textit{run}) $\Sigma$ is a sequence of events,
where every event is associated with a distinct process and
every \textit{receive}($m$) event has a preceding matching \textit{send}($m$) event. 
A process \textit{misbehaves} in a run (we also call it \emph{Byzantine}) if it produces an event that is not prescribed by the assigned protocol, given the preceding sequence of events, starting from the initial configuration $C^0$. 
If a process does not misbehave, we call it \emph{benign}.
In an infinite run, a process \textit{crashes} if it prematurely stops producing events required by the protocol; 
if a process is benign and never crashes we call it \emph{correct}, and it is 
\emph{faulty} otherwise. 
Let $\textit{part}(\Sigma)$ denote the set of processes that produce events in an execution $\Sigma$.

\subsection{Channels and digital signatures}

Every pair of processes communicate over a \textit{reliable channel}: in every infinite run, if a correct process $p$ sends a message $m$ to a correct process $q$, $m$ eventually arrives, and $q$ receives a message from $p$ only if $p$ sent it.
We impose \emph{no synchrony assumptions}. 
In particular, we assume no bounds on the time required to convey a message from one correct process to another. 
In the following, we assume that all messages sent with a protocol execution are \emph{signed}, and the signatures can be \emph{verified} by a third party.
In particular, each time a process $p$ receives a protocol message $m$ from process $q$, $p$ only accepts $m$ if it is properly signed by $q$. 
We assume a computationally bound \emph{adversary}: no process can forge the signature of a benign process.

\subsection{Decentralized trust}
We now formally define our decentralized trust assumptions.  
A \emph{quorum map} $\cQ: \Pi \rightarrow 2^{2^{\Pi}}$ provides every process with a set of process subsets: 
for every process $p$, $\cQ(p)$ is the set of \emph{quorums of $p$}.  
We assume that $p$ includes itself in each of its quorums: $\forall Q\in \cQ(p): p\in Q$.
Intuitively, $\cQ(p)$ describes what process $p$ \emph{expects} from the system. 
We implicitly assume that, from $p$'s perspective, for every quorum  $Q\in\cQ(p)$, there is an execution in which $Q$ is precisely the set of correct processes. 
However, these expectations may be violated by the environment.
We therefore introduce a \emph{fault model} $\cF\subseteq 2^{\Pi}$ (sometimes also called an \emph{adversary structure}) stipulating which process subsets can be faulty.   
In this paper, we assume \emph{inclusion-closed} fault models that do not force processes to fail: $\forall F\in \cF,\; F'\subseteq F: \; F'\in \cF$. 
An execution $\Sigma$ \emph{complies with $\cF$} if the set of faulty processes in $\Sigma$ is in $\cF$. 

Given a faulty set $F\in \cF$, a process $p$ is called \emph{live in $F$} if it has a \emph{live quorum in $F$}, i.e., $\exists Q\in \cQ(p): Q \cap F = \emptyset$.
Intuitively, if $p$ is live in every $F\in \cF$, such that $p\notin F$, then its trust assumptions are justified by the environment. 

For example, let the uniform \emph{$f$-resilient} fault model: $\cF=\{F\subseteq \Pi: |F|\leq f\}$.
If $\cQ(p)$ includes all sets of $n-f$ processes, then $p$ is guaranteed to have at least one live quorum in every execution.
On the other hand, if $\cQ(p)$ expects that a selected process $q$ is always correct ($q\in \cap_{Q\in\cQ(p)} Q$), then $p$ is not live in any execution with a faulty set such that $q\in F$. %
In the rest of the paper, we assume that the model is provided with \emph{trust assumptions} $(\cQ,\cF)$, where $\cQ$ is a quorum map and $\cF$ is a fault model.

\section{The Broadcast Primitive} 
\label{sec:b_protocol}

The broadcast abstraction exports input events $\textit{broadcast}(m)$ and output events $\textit{deliver}(m)$, for value $m$ in a \emph{value set} $\cM$. 
Without loss of generality, we assume that each broadcast instance  has a dedicated \textit{source}, i.e., the process invoking the $\textit{broadcast}$ operation.\footnote{One can easily turn this (one-shot) abstraction into a \emph{long-lived} one, in which every process can broadcast arbitrarily many distinct values by equipping each broadcast value with a source identifier and a unique \emph{sequence number} and its signature.}
Below we recall the classical abstractions of  consistent and reliable broadcast \cite{cachin2011introduction}. 
The \emph{consistent broadcast} abstraction satisfies:

\begin{itemize}
    \item (Validity) If the source is correct and broadcasts $m$, then every correct process eventually delivers $m$.
    \item (Consistency) If correct processes $p$ and $q$ deliver $m$ and $m'$ respectively, then $m = m'$.
    \item (Integrity) Every correct process delivers at most one value and, if the source $p$ is correct, only if $p$ previously broadcast it.
\end{itemize}

A reliable broadcast protocol satisfies the properties above, plus:

\begin{itemize}
    \item (Totality) If a correct process delivers a value, then all correct processes eventually deliver a value.
\end{itemize}

For our lower bound, we introduce a relaxed version of consistent broadcast. 
A \textit{$k$-consistent broadcast protocol} ($k$-CB) ensures that in every execution $\Sigma$ (where $F \in \mathcal{F}$ is its faulty set), the following properties are satisfied:

\begin{itemize}
    \item (Validity) If the source is correct and broadcasts $m$, then every \emph{live} correct process eventually delivers $m$.
    \item ($k$-Consistency) Let $M$ be the set of values delivered by the correct processes, then $|M| \leq k$.
    \item (Integrity) A correct process delivers at most one value and, if the source $p$ is correct, only if $p$ previously broadcast it.
\end{itemize}

In this paper, we restrict our attention on \emph{quorum-based} protocols~\cite{losa2019stellar}.
Intuitively, in a quorum-based protocol, every process $p$ is expected to make progress if the members of one of its quorums $Q\in\cQ(p)$ appear correct to $p$.
This should hold even if the actual set of correct processes in this execution is different from $Q$.
The property has been originally introduced in the context of consensus protocols~\cite{losa2019stellar}. Here we extend it to broadcast.
Formally, we introduce the following property that completes the specification of  $k$-CB:

\begin{itemize}
    \item (Local Progress) For all $p \in \Pi$ and $Q \in \cQ(p)$, 
    there is an execution in which only the source and  processes in $Q$ take steps, $p$ is correct, and $p$ delivers a value.
\end{itemize}

The key differences of a \textit{k}-CB over a classical consistent broadcast lies in the Validity and the \textit{k}-Consistency properties. 
Our Validity property only ensures progress to \emph{live} correct processes (based on their local quorums). 
Also, since some processes may trust the "wrong guys", it might happen that a faulty source convinces the correct processes to deliver distinct values.
However, given a fault model $\mathcal{F}$, the \textit{k}-Consistency property establishes an upper bound $k$ in values that can be delivered. 
In the classical consistent broadcast, no conflict is allowed in values delivered for a given $\mathcal{F}$, the bound $k$ on such primitive is then equal to $1$ (which clearly also holds for reliable broadcast).

\section{Bounds for \textit{k}-consistent broadcast protocol}
\label{sec:lower}

\subsection{A Graph Representation of Executions}

We use our trust assumptions $(\cQ,\cF)$ to build a graph representation of the execution, in order to investigate the cases in which disagreement may occur in the network, that is, when two or more correct processes deliver distinct values. 
Let $S: \Pi \rightarrow 2^{\Pi}$ be a map providing each process with one of its quorums, that is, $S(p) \in \cQ(p)$. 
Let $\cS$ be the family of all possible such maps $S$.

Given $F \in \mathcal{F}$ and $S \in \cS$, we build an undirected graph $G_{F,S}$ as follows: 
the nodes in $G_{F,S}$ are correct processes ($\Pi-F$) and the edges are drawn between a pair of nodes if their quorums intersect in at least one correct process. 
Formally, $G_{F,S}$ is a pair $(\Pi_{F},E_{F,S})$ in which:

\begin{itemize}
    \item $\Pi_{F} = \Pi - F$
    \item $(p, q) \in E_{F,S} \Leftrightarrow S(p) \cap S(q) \not\subseteq F$
\end{itemize}

\begin{example}
\label{ex:graphs}
Let us consider the system where $\Pi = \{p1,p2,p3,p4\}$, given the faulty set $F = \{p3\}$ and the quorum system for each process:
\[ \cQ(p1) = \{\{p1,p2,p3\},\{p1,p3,p4\}\} \ \ \ \  \cQ(p2) = \{\{p1,p2,p3\},\{p2,p3,p4\}\} \]\[\cQ(p3) = \{\{p1,p2,p4\},\{p2,p3,p4\}\} \ \ \ \ \cQ(p4)  = \{\{p1,p3,p4\},\{p2,p4\},\{p3,p4\}\} \]

If we consider only the correct processes ($p1$, $p2$ and $p4$), there are $12$ different combinations of quorums $S \in \cS$ for these trust assumptions. Now let $S_1 \in \cS$ with: $S_1(p1) = \{p1,p2,p3\}$, $S_1(p2) = \{p2,p3,p4\}$ and $S_1(p4) = \{p2,p4\}$. And let $S_2 \in \cS$ with: $S_2(p1) = \{p1,p2,p3\}$, $S_2(p2) = \{p2,p3,p4\}$ and $S_2(p4) = \{p3,p4\}$. Figure \ref{ex1_graphs} shows the graphs $G_{F,S_1}$ and $G_{F,S_2}$, observe that every pair of quorums used to generate $G_{F,S_1}$ intersects in a correct process, thus resulting in a fully connected graph. On the other hand, since $S_2(p1) \cap S_2(p4) \subseteq F$, $G_{F,S_2}$ is not fully connected.
\end{example}

\begin{figure}
\centering
\includegraphics[scale=0.6]{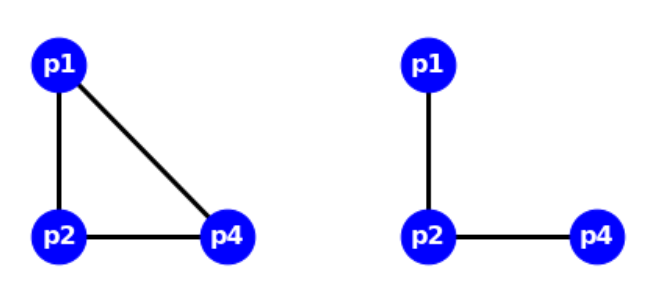}
\caption{Graph structures of Example~\ref{ex:graphs}: $G_{F,S_1}$ and $G_{F,S_2}$ respectively.}
\label{ex1_graphs}
\end{figure}

Before proceeding with the analysis, we recall the following classical definitions:

\begin{definition}[Independent Set] 

A set of nodes $I$ is an \emph{independent set} of a graph $G_{F,S}$ if no pair of nodes in $I$ is adjacent, i.e., $\forall p,q \in I: (p,q) \notin E_{F,S}$.

\end{definition}

\begin{definition}[Independence Number]
The \emph{independence number} of a graph $G_{F,S}$ is the size of its largest independent set.
\end{definition}

Within an independent set of $G_{F,S}$, the quorums of each pair of nodes do not intersect in a correct process. 
The independence number of $G_{F,S}$ helps us in understanding the level of disagreement that might occur in an execution, as we show in the following:

\begin{theorem}
\label{th:n_conflict}
Let $G_{F,S}$ be the graph generated over a fixed $F \in \mathcal{F}$ and $S \in \cS$. 
Let $G_{F,S}$ have an independent set of size $k$. 
Then there exists an execution 
in which up to $k$ distinct values can be delivered by correct processes.
\end{theorem}

\begin{proof}
Let $r$ be the source. 
If $I = \{p_{1},...,p_{k}\}$ is an independent set of $G_{F,S}$ of size $k$,
then $\forall p_i,p_j \in I: S(p_i) \cap S(p_j) \subseteq F$.
By the definition of Local Progress, it exists an execution $\Sigma_i$ such that $part(\Sigma_i)= \{r\} \cup S(p_i)$ and $p_i$ delivers a value $m_i$. It then suffices for $r$ and other faulty processes in $S(p_i)$ to behave exactly as they do within $\Sigma_i$ in order to produce the same result.
Since the system is asynchronous, it is possible that $p_i$ delivers a value before any 
correct
process in \textit{part}($\Sigma_i$) receives a message from any $p' \notin \textit{part}(\Sigma_i)$.
In other words, the network behaves as if it was temporarily partitioned.
Now for each $p_{i} \in I$, let $\Sigma_i$ be an execution as described above, we can build $\Sigma$ such that all executions $\Sigma_i$ are subsequences of $\Sigma$, in which no correct process receives any information of conflicting values before $p_{1},...,p_{k}$ deliver $m_{1},...,m_{k}$, respectively.
\end{proof}

\begin{example}
\label{ex:indepedence}
Coming back to Example~\ref{ex:graphs}, we see that the nodes in $G_{F,S_1}$ are fully connected (form a clique), thus resulting in $G_{F,S_1}$ having independence number $1$. 
On the other hand, the biggest independent set in $G_{F,S_2}$ is $\{p1,p4\}$, which means $G_{F,S_2}$ has independence number $2$. 
In an execution where $F$ is the faulty set and processes first hear from quorums in $S_2$ to deliver a value, 
then there is an unavoidable possibility that $p1$ and $p4$ deliver distinct values if $p3$ is the source.
\end{example}

\subsection{Lower bound on \textit{k}}

Given  the pair $(\cQ,\cF)$, 
we define the family of graphs $\cG_{\cQ,\cF}$ that includes all possible $G_{F,S}$, where $F \in \mathcal{F}$ and $S \in \cS$.
Recall that every such $S \in \cS$ associates each  
process to one of its quorums.

\begin{definition}[Inconsistency Number]
Let $\mu:\cG_{\cQ,\cF} \rightarrow N$ map each $G_{F,S} \in \cG_{\cQ,\cF}$ to its independence number.
The \emph{inconsistency number} of $(\cQ,\cF)$ is then $k_{max} = max(\{\mu(G_{F,S})|G_{F,S} \in \cG_{\cQ,\cF}\})$.
\end{definition}

\begin{theorem}
\label{th:allbound}
No algorithm can implement $k$-CB with $k < k_{max}$.
\end{theorem}

\begin{proof}
For a particular $G_{F,S} \in \cG_{\cQ,\cF}$, Theorem~\ref{th:n_conflict} implies that within an independent set $I:|I| = k$, up to $k$ distinct values can be delivered by the correct processes. 
As the independence number is the size of the maximum independent set(s) of a graph, by taking the highest independence number in $\cG_{\cQ,\cF}$ we get the worst case scenario.
It is always possible to build an execution where $k_{max}$ processes deliver $k_{max}$ distinct values before any correct process is able to identify the misbehavior.
\end{proof}

\begin{example}
\label{ex:optimal}
Coming back to Example~\ref{ex:graphs}  again, if we take $S_3$ such that $S_3(p1) = S_3(p2) = \{p1,p2,p3\}$ and $S_3(p4) = \{p3,p4\}$, we have both $(p1,p4) \not\in E_{F,S_3}$ and $(p2,p4) \not\in E_{F,S_3}$, while $(p1,p2) \in E_{F,S_3}$. The independence number of $G_{F,S_3}$ is 2, which means that despite $G_{F,S_2}$ having more edges then $G_{F,S_3}$, the same number of distinct values can be delivered by correct processes. For $\mathcal{F} = \{\{p3\}\}$, since the quorums of $p1$ and $p2$ always intersect on a correct process, none of the graphs has independence number higher then $G_{F,S_3}$, thus, considering $\cQ$ from Example~\ref{ex:graphs} and $\mathcal{F}$, the optimal $k$ for an algorithm implementing $k$-CB would be 2.
\end{example}

\section{Accountable Algorithm for Relaxed Broadcast}
\label{sec:algo}

In the specification of \textit{k}-CB, we inherently assume the possibility of correct processes disagreeing in the delivered value in the presence of a faulty source, but the maximal number of distinct delivered values is determined by $(\cQ,\cF)$. 

In practice, one may need some form of Totality, as in  reliable broadcast.
We might want the (live) correct processes to reach some form of agreement on the set of values they deliver. 

In our setting, we have to define the Totality property, taking into account the possibility of them delivering different values, in case the source is misbehaving.
Therefore we strengthen the protocol by adding an accountability feature: 
once a correct process detects misbehavior of the source, i.e., it finds out that the source signed two different values, it can use the two signatures as a \emph{proof of misbehavior}.
The proof can be then independently verified by a third party.
We model the accusation as an additional output \textit{accuse($mb$)}, where $mb$ is a proof that the source misbehaved. 
When a process $p$ produces \textit{accuse($mb$)}, we say that $p$ \emph{accuses the source} (of misbehavior with proof $mb$).   
Now, in addition to the properties of  \textit{k}-CB, the \textit{k-consistent reliable broadcast} (\textit{k}-CRB) abstraction satisfies:

\begin{itemize}
    \item (Weak Totality) If a correct process delivers a value, then every \emph{live} correct process eventually delivers a value or accuses the source.
    \item (Accuracy) A correct process $p$ accuses the source only if the source is faulty.
    \item (Certitude) If a correct process accuses the source, every correct process eventually does so.
\end{itemize}

We present our \textit{k}-CRB implementation in Algorithm~\ref{alg:1phase}.
Each process maintains local variables \textit{sentecho}, \textit{delivered}, \textit{accused} and \textit{echoes}.
Boolean variables \textit{sentecho}, \textit{delivered} and \textit{accused} indicate whether $p_{i}$ has already sent \textit{ECHO}, delivered a value and accused the source, resp., in the broadcast instance.
Array \textit{echoes} keeps track of \textit{ECHO} messages received from other processes. 

The source broadcasts $m$ by sending a \textit{SEND} message to every process in the system. 
If a process $p_{i}$ receives either a [\textit{SEND},$m$] or a [\textit{ECHO},$m$] for the first time, $p_{i}$ sends an \textit{ECHO} message to every other processes. 
If a received \textit{ECHO} message contains a value $m_2$ that conflicts with a previously received value $m_1$, $p_i$ sends the \textit{ACC} message to every process with the tuple $(m_1,m_2)$ as a proof of misbehavior. 
Once $p_i$ receives echoes with $m$ from at least one of its quorum, it  delivers $m$. 
Once $p_i$ receives an \textit{ACC} message containing a proof of misbehavior, even though $p_i$ has already delivered a value, it also sends \textit{ACC} to every process before accusing the source. 
Notice that a correct process only sends \textit{ECHO} for a single value, and delivers a value or accuses the source once.

\begin{algorithm}
\SetAlgoLined
\BlankLine
\textbf{Local Variables:} \\
\textit{sentecho $\leftarrow FALSE$;} $\setminus$$\setminus$Indicate if $p_{i}$ has sent \textit{ECHO} \\
\textit{delivered $\leftarrow FALSE$;} $\setminus$$\setminus$Indicate if $p_{i}$ has delivered a value \\
\textit{accused $\leftarrow FALSE$;} $\setminus$$\setminus$Indicate if $p_{i}$ has accused the source \\
\textit{echoes $\leftarrow [\perp]^{N}$;} $\setminus$$\setminus$Array of received \textit{ECHO} messages from others processes \\
\BlankLine
\textbf{upon invoking broadcast($m$)}: \{ If $p_i$ is the source \}\\
    \ \ \ \ \textit{send message [\textit{SEND},$m$] to all $p_{j} \in \Pi$;} \\
\BlankLine    
\textbf{upon receiving a message [\textit{SEND},$m$] from $p_{j}$:} \\
    \ \ \ \ \textit{\textbf{if}($\neg$\textit{sentecho}):} \\
    \ \ \ \ \ \ \ \ \textit{sentecho $\leftarrow TRUE$;} \\
    \ \ \ \ \ \ \ \ \textit{send message [\textit{ECHO},$m$] to all $p_{j} \in \Pi$;} \\
\BlankLine   
\textbf{upon receiving a message [\textit{ECHO},$m$] from $p_{j}$:} \\
    \ \ \ \ \textit{echoes[j] $\leftarrow$ $m$;} \\
    \ \ \ \ \textit{\textbf{if}(there exists \textit{echoes}[$k$] $\neq \perp$ such that \textit{echoes}[$k$] $\neq$ \textit{echoes}[j]):} \\
     \ \ \ \ \ \ \ \ \textit{$m1$ $\leftarrow$ \textit{echoes}[j];} \\
     \ \ \ \ \ \ \ \ \textit{$m2$ $\leftarrow$ \textit{echoes}[k];} \\
    \ \ \ \ \ \ \ \ \textit{send message [\textit{ACC},$(m1,m2)$] to all $p_{j} \in \Pi$;} \\
    \ \ \ \ \ \ \ \ \textit{accuse $(m1,m2)$;} \\
    \ \ \ \ \textit{\textbf{if}($\neg$\textit{sentecho}):} \\
    \ \ \ \ \ \ \ \ \textit{sentecho $\leftarrow TRUE$;} \\
    \ \ \ \ \ \ \ \ \textit{send message [\textit{ECHO},$m$] to all $p_{j} \in \Pi$;} \\
\BlankLine
\textbf{upon receiving a message [\textit{ACC},$(m1,m2)$] from $p_{j}$:} \\
    \ \ \ \ \textit{\textbf{if}($\neg \textit{accused}$)}\\
        \ \ \ \ \ \textit{accused $\leftarrow TRUE$;} \\
    \ \ \ \ \ \ \textit{send message [\textit{ACC},$(m1,m2)$] to all $p_{j} \in \Pi$;} \\
    \ \ \ \ \ \ \textit{accuse $(m1,m2)$;}
\BlankLine
\textbf{upon receiving \textit{ECHO} for $m$ from every $q \in Q_{i}, Q_{i} \in \cQ(p_{i})$:} \\
    \ \ \ \ \textit{\textbf{if}($\neg \textit{delivered}$)}\\
    \ \ \ \ \ \textit{delivered $\leftarrow TRUE$;} \\
    \ \ \ \ \ \textit{deliver $m$;} \\
\BlankLine
\caption{1-Phase Broadcast Algorithm: code for process $p_{i}$}
\label{alg:1phase}
\end{algorithm}

Process $p_i$ delivers a value $m$ after receiving \textit{ECHO} from every process in $Q_i \in \cQ(p_i)$, we say that $p_i$ \emph{uses $Q_i$}. 
In our correctness arguments, we fix an execution of Algorithm \ref{alg:1phase} with a faulty set $F \in \mathcal{F}$, and assume that the processes use quorums defined by a fixed map $S \in \cS$.

\begin{lemma}
\label{algo_pairwise}
Let $G_{F,S}$ be the graph generated over $F$ and $S$ with $(p,q) \in E_{F,S}$, if $p$ delivers $m_1$ and $q$ delivers $m_2$, then $m_1 = m_2$. 
\end{lemma}

\begin{proof}
Since $p$ delivers $m_1$ using $S(p)$, all processes in $S(p)$ sent \textit{ECHO} with $m_1$ to $p$. Similarly, all processes in $S(q)$ sent \textit{ECHO} with $m_2$ to $q$. Assume that $m_1 \neq m_2$, since $(p,q) \in E_{F,S} \Leftrightarrow S(p) \cap S(q) \not\subseteq F$, some correct process sent \textit{ECHO} with $m_1$ and $m_2$, which is not allowed by the protocol.
\end{proof}

As an immediate consequence of Lemma~\ref{algo_pairwise}, correct processes $p$ and $q$ might deliver distinct values only if $(p,q) \notin E_{F,S}$. 

\begin{theorem}
\label{algo_th:n_conflict}
Let $k$ be the independence number of $G_{F,S}$, then $k$ is an upper bound in the number of distinct values that can be delivered by correct processes.
\end{theorem}

\begin{proof}
Lemma~\ref{algo_pairwise} states that if the quorums of correct processes intersect in a correct process, they cannot deliver conflicting values using those quorums. Let $I$ be an independent set in $G_{F,S}$ of size $k$ and assume that more than $k$ distinct values are delivered, then for some $q \notin I, \exists p \in I: (p,q) \in E_{F,S}$, in which $p$ and $q$ deliver distinct values, a contradiction.
\end{proof}

\begin{theorem}
\label{theor:reliable_alg}
Consider a distributed system with trust assumptions $(\cQ,\cF)$. Let 
$k_{max}$ be the inconsistency number of $(\cQ,\cF)$. 
Then Algorithm~\ref{alg:1phase} implements $k_{max}$-consistent reliable broadcast.
\end{theorem}

\begin{proof}
(Integrity) 
Immediate from the algorithm: a process only delivers a value once, and if the source is correct and broadcasts $m$, no process can deliver a value different from $m$. 

(Accuracy) 
A correct process accuses the source after receiving echoes with distinct values or an accusation from another process.
In both cases, the signature of the source is verified.
Correct processes do not broadcast distinct values, and since a faulty process cannot forge signatures of a correct one, it follows that distinct values can only come from a faulty source.

(Certitude) In both situations in which a correct process accuses misbehavior, it previously sends an \textit{ACC} message to every process in the network containing a pair of distinct values. 
The message is then eventually received by every correct process in the network, which accuses the source as well.

(Validity) When a correct process broadcasts $m$, it sends [\textit{SEND},$m$] to every process in the network. 
Every correct process eventually receives the message and echoes it to every process. 
If a correct process has a live quorum $Q$, it will eventually receive \textit{ECHO} with $m$ from all the processes in $Q$ and deliver the value.

(Weak Totality) A correct process $p$ sends an \textit{ECHO} message to every process after receiving it if $p$ has not previously echoed a value. 
Consequently, if some correct process receives an \textit{ECHO} message, every correct process eventually does so. 
If $p$ delivers a value, it must have received at least one \textit{ECHO} message,
in which case, every correct process eventually receives and echoes a value.
If a correct process $q$ has a live quorum $Q$, it eventually receives \textit{ECHO} from all processes in $Q$. 
Two cases are then possible.
If all of the \textit{ECHO} messages received by $q$ contains the same value, then $q$ delivers it.
Otherwise, $q$ accuses misbehavior.

($k_{max}$-Consistency) 
Let $G'_{F,S}$ be the graph whose independence number is $k_{max}$. %
By Theorem~\ref{algo_th:n_conflict}, the number of distinct values that can be delivered by correct processes in a given execution is bounded by $k_{max}$. 
As, by definition, $k_{max}$ is the higher independence number of graphs in $\cG_{\cQ,\cF}$, \textit{$k_{max}$-consistency} is ensured.
\end{proof}

An algorithm implementing $k$-CRB satisfies the required properties of $k$-CB, thus, it also implements $k$-CB. From Theorem~\ref{th:allbound}, no algorithm can implement $k$-CB with $k < k_{max}$, therefore, Theorem \ref{theor:reliable_alg}~implies that Algorithm \ref{alg:1phase} implements $k$-CB, and consequently $k$-CRB, with optimal $k$.

\noindent
\textbf{Computing inconsistency parameters.}
A straightforward approach to find the inconsistency number of $(\cQ,\cF)$ consists in computing the independence number of all graphs $G_{F,S} \in \cG_{\cQ,\cF}$. 
The problem of finding the largest independent set in a graph (called \textit{maximum independent set}), and consequently its independence number, is the \textit{maximum independent set problem} \cite{tarjan1977finding},
known to be \textit{NP-complete} \cite{miller2013complexity}.
Also, the number of graphs in $\cG_{\cQ,\cF}$ may exponentially grow with the number of processes.
However, as the graphs might have similar structures (for example, the same quorums for some processes may appear in multiple graphs),
in many practical scenarios, we should be able to avoid redundant calculations and reduce the overall computational costs.

\section{Related Work}
\label{sec:related}

Assuming synchronous communication,  Damg{\aa}rd et al.~\cite{damgaard2007secure} described protocols implementing broadcast, verifiable secret sharing and multiparty computation in the decentralized trust setting. They introduce the notion of  \textit{aggregate adversary structure}~$\mathcal{A}$: each node is assigned a collection of subsets of nodes that the adversary might corrupt at once. 

Ripple~\cite{schwartz2014ripple} is arguably the first practical partially synchronous system based on  decentralized trust assumptions. 
In the Ripple protocol, each participant express its trust assumptions in the form of an \textit{unique node list} (UNL), a subset of nodes of the network. In order to accept transactions, a node needs to "hear" from at least $80\%$ of its UNL, and according to the original white paper \cite{schwartz2014ripple}, assuming that up to $20\%$ of the nodes in an UNL might be Byzantine, the overlap between every pair of UNL's needed to prevent forks was believed to be $\geq 20\%$.
The original protocol description appeared to be sketchy and informal, and later works detailed the functioning of the protocol and helped to clarify under which conditions its \textit{safety} and \textit{liveness} properties hold \cite{armknecht2015ripple,chase2018analysis,mauri2020formal,amores2020security}.
In particular, it has been spotted~\cite{armknecht2015ripple} that  its safety properties can be violated (a \emph{fork} can happen) with as little as $20\%$ of UNLs overlap, even if there are no Byzantine nodes. 
It then establishes an overlap bound of $>40\%$ to guarantee consistency without Byzantine faults. 
In a further analysis, assuming that at most $20\%$ of nodes in the UNLs are Byzantine, \cite{chase2018analysis} suggests an overlap of $>90\%$ in order to prevent forks, but also provide an example in which the liveness of the protocol is violated even with $99\%$ of overlap. 
Recently, a formalization of the algorithm was presented in \cite{amores2020security}, and a better analysis of the correctness of the protocol in the light of an \textit{atomic broadcast} abstraction was given by Amores-Cesar et al.~\cite{amores2020security}.

The Stellar consensus protocol~\cite{mazieres2015stellar} introduces the \textit{Federated Byzantine Quorum System} (FBQS).
A quorum $Q$ in the FQBS is a set that includes a \emph{quorum slice} (a trusted subset of nodes) for every node in $Q$. 
Correctness of Stellar depends on the individual trust assumptions and are only guaranteed for nodes in the so called \textit{intact set}, which is, informally, a set of nodes trusting the "right guys".
Garc\'ia-P\'erez and Gotsman~\cite{garcia2018federated} formally argue about Stellar consensus, by relating it to Bracha's Broadcast Protocol~\cite{bracha1987asynchronous}, build on top of a FBQS.
The analysis has been later extended~\cite{garcia2019deconstructing} to a variant of state-machine replication protocol that allows \emph{forks}, where 
disjoint intact sets may maintain different copies of the system state.

Cachin and Tackmann~\cite{CT19} defined the notion of \textit{Asymmetric Quorum Systems}, based on individual adversary structures. 
They introduced a variant of broadcast whose correctness is restricted to a \emph{guild}, a subset of nodes that, similarly to the intact nodes in the Stellar protocol, have the "right" trust assumptions. Executions with a guild also ensure consistency (correct processes do not deliver distinct values). In our approach, we relax the consistency property, allowing for more flexible trust assumptions, while using accountability to ensure correctness for every live correct process.

In the similar vein, Losa et al.~\cite{losa2019stellar}, define the quorum system used by Stellar using the notion of a \textit{Personal Byzantine Quorum System} (PBQS), where every process chooses its quorums with the restriction that if $Q$ is a quorum for a process $p$, then $Q$ includes a quorum for every process $q' \in Q$. 
They show that for any quorum-based algorithm  (close to what we call an algorithm satisfying the Local Progress property), consensus is not achievable in partially synchronous systems where two processes have quorums not intersecting on a correct process.
The paper also determines the conditions under which a subset of processes can locally maintain safety and liveness, even though the system might not be globally consistent. We use a similar approach in the context of broadcast, and in addition to a relaxed consistency guarantee, we also parameterize the level of disagreement in the network using the individual trust assumptions.

In the context of distributed systems, accountability has been proposed as a mechanism to detect ``observable'' deviations of system nodes from the algorithms they are assigned with~\cite{detection-case,detection-problem,peerreview}.      
Recent proposals~\cite{polygraph,rala} focus on \emph{application-specific} accountability that only heads for detecting misbehavior that affects correctness of the problem to be solved, e.g., consensus~\cite{polygraph} or lattice agreement~\cite{rala}.   
Our $k$-CRB algorithm generally follows this approach, except that it implements a \emph{relaxed} form of broadcast, but detects violations that affect correctness of the stronger, conventional reliable broadcast~\cite{cachin2011introduction}.

\section{Concluding Remarks}
\label{sec:conclusion}

In this paper, we address a  realistic scenario in which correct processes choose their trust assumptions in a purely decentralized way.
The resulting structure of their trust relations may cause inevitable violations of consistency properties of conventional broadcast definitions.
Our goal is to precisely quantify this inconsistency by considering relaxed broadcast definitions: $k$-consistent broadcast and $k$-consistent reliable broadcast.

In case the broadcast source is Byzantine, the abstractions allow correct processes to deliver up to $k$ different values.
We show that $k$, the optimal ``measure of inconsistency'', is the highest independence number over all graphs $G_{F,S}$ in a family $\cG_{\cQ,\cF}$ determined by the given trust assumptions $(\cQ,\cF)$.
We show that this optimal $k$ can be achieved by a \emph{$k$-consistent reliable broadcast} protocol that, in addition to $k$-consistency also provides a form of accountability: if a correct process delivers a value, then every live correct process either delivers some value or detects the source to be Byzantine. 

A natural question for the future work is to quantify inconsistency in higher-level abstractions, such as distributed storage or  asset-transfer systems~\cite{cons-crypto,astro-dsn} that can be built atop the relaxed broadcast abstractions. 
Another interesting direction would be in self-reconfigurable systems~\cite{rala}: since we expect the system to admit disagreement, once a Byzantine process is detected, other participants may want to update their trust assumptions.
It is also extremely appealing to generalize the very notion of a quorum system to \emph{weighted} quorums, where the contribution of a quorum member is proportional to its \emph{stake} in an asset transfer system~\cite{pastro-tr-21}.
This opens a way towards \emph{permissionless} asset transfer systems with relaxed guarantees.
    
\bibliographystyle{abbrv}

\bibliography{main.bbl}

\end{document}